\documentclass[12pt]{article}%
\usepackage[letterpaper, top=1in, bottom=1in, left=1in, right=1in]{geometry}
\usepackage{amsfonts, amsmath, amsthm, amssymb}
\usepackage{enumitem, graphicx, stmaryrd, color, bbm}
\usepackage{url, algorithm}
\usepackage{dsfont}
\usepackage{float}
\floatstyle{ruled}

\usepackage[driverfallback=hypertex,pagebackref=true,colorlinks]{hyperref}
	\hypersetup{linkcolor=[rgb]{.7,0,.7}}
	\hypersetup{citecolor=[rgb]{.5,0,.5}}
	\hypersetup{urlcolor=[rgb]{.7,0,.7}}

\newtheorem{theorem}{Theorem}
\newtheorem{lemma}[theorem]{Lemma}

\newtheorem{definition}[theorem]{Definition}

\newtheorem{remark}[theorem]{Remark}

\newtheorem{proposition}[theorem]{Proposition}
\newtheorem{example}[theorem]{Example}

\newtheorem{fact}[theorem]{Fact}

\newcommand{\abs}[1]{\left|#1\right|}		% absolute value
\newcommand{\E}{\mathop{\mathbb{E}}}  		% expectation
\newcommand{\R}{\mathop{\mathbb{R}}}  		% real numbers
\newcommand{\N}{\mathbb{N}} 			  		% natural numbers
\newcommand{\F}{\mathbb{F}}					% fields
\newcommand{\Z}{\mathbb{Z}}					% integers
\newcommand{\C}{\mathbb{C}}					% complex numbers

\newcommand{\mc}{\mathcal}

\newcommand{\eps}{\mathop{\epsilon}}
\newcommand{\set}[1]{\left\{ #1 \right\}}   % sets
\newcommand{\brac}[1]{\left( #1 \right)}    % brackets
\newcommand{\sqbrac}[1]{\left[ #1 \right]}  % square brackets

\newcommand{\Var}{\textnormal{Var}}				% variance
\newcommand{\up}[1]{^{\brac{#1}}}

\newcommand{\ind}{\mathds{1}}

\newcommand{\supp}{\textnormal{supp}} 	
\newcommand{\modt}{\textnormal{ (mod 2)}} 
\newcommand{\modk}{\textnormal{ (mod $k-1$)}} 
\newcommand{\Lin}{\textnormal{Lin}} 
\newcommand{\Sym}{\textnormal{Sym}}  

\let\leq=\leqslant % redefine them
\let\geq=\geqslant %

\begin{document}

%\sloppy

\title{Biased Linearity Testing in the 1\% Regime}
%\author{}
\author{Subhash Khot\thanks{Department of Computer Science, Courant Institute of Mathematical Sciences, New York University. E-mail: \href{khot@cs.nyu.edu}{\texttt{khot@cs.nyu.edu}.} Research supported by NSF Award CCF-1422159, NSF Award CCF-2130816, and the Simons Investigator Award.} \and Kunal Mittal\thanks{Department of Computer Science, Princeton University. E-mail: \href{kmittal@cs.princeton.edu}{\texttt{kmittal@cs.princeton.edu}.} Research supported by NSF Award CCF-2007462, and the Simons Investigator Award.}}
\date{\today}			 					% \today for date now
\maketitle
%\vspace{-1em}
\begin{abstract}
We study linearity testing over the $p$-biased hypercube $(\{0,1\}^n, \mu_p^{\otimes n})$ in the 1\% regime.
For a distribution $\nu$ supported over $\{x\in \{0,1\}^k:\sum_{i=1}^k x_i=0 \textnormal{ (mod 2)} \}$, with marginal distribution $\mu_p$ in each coordinate, the corresponding $k$-query linearity test $\textnormal{Lin}(\nu)$ proceeds as follows:
Given query access to a function $f:\{0,1\}^n\to \{-1,1\}$, sample $(x_1,\dots,x_k)\sim \nu^{\otimes n}$, query $f$ on $x_1,\dots,x_k$,  and accept if and only if $\prod_{i\in [k]}f(x_i)=1$.
	
Building on the work of Bhangale, Khot, and Minzer (STOC '23), we show, for $0<p\leq \frac{1}{2}$, that if $k \geq  1+\frac{1}{p}$, then there exists a distribution $\nu$ such that the test $\textnormal{Lin}(\nu)$ works in the 1\% regime; that is, any function $f:\{0,1\}^n\to \{-1,1\}$ passing the test $\textnormal{Lin}(\nu)$ with probability $\geq \frac{1}{2}+\epsilon$, for some constant $\epsilon>0$, satisfies $\Pr_{x\sim \mu_p^{\otimes n}}[f(x)=g(x)] \geq \frac{1}{2}+\delta$, for some linear function $g$, and a constant $\delta = \delta(\epsilon)>0$.
	
Conversely, we show that if $k < 1+\frac{1}{p}$, then no such test $\textnormal{Lin}(\nu)$ works in the 1\% regime. 
Our key observation is that the linearity test $\textnormal{Lin}(\nu)$ works if and only if the distribution $\nu$ satisfies a certain pairwise independence property.

\end{abstract}

\section{Introduction}

A function $f:\set{0,1}^n\to \set{-1,1}$ is said to be linear over $\F_2$\footnote{by identifying the range $\F_2$ with $\set{-1,1}$, under the map $b\to (-1)^b$}, if there exists a set $S\subseteq [n]$, such that $f(x) = \prod_{i\in S} \brac{-1}^{x_i}$; this function is denoted by $\chi_S$.
The classical linearity testing problem, asks, given query access\footnote{the algorithm is allowed to ask/query the value of $f(x)$ at any $x\in \set{0,1}^n$} to a function $f:\set{0,1}^n\to \set{-1,1}$, to distinguish between the following two cases\footnote{the algorithm is allowed to answer arbitrarily for functions $f$ which violate both the conditions}:
\begin{enumerate}
	\item $f$ is a linear function.
	\item $f$ is far from being linear; that is, for every linear function $\chi_S$, the functions $f$ and $\chi_S$ disagree on many points.
\end{enumerate}

Linearity testing was first studied by Blum, Luby, and Rubinfeld, who gave a very simple 3-query test for this problem~\cite{BLR93}.
This test, known as the BLR test, proceeds in the following manner: Sample $x, y \sim \set{0,1}^n$ uniformly and independently; query $f$ at $x, y,$ and $x\oplus y$, and accept if and only if $f(x\oplus y)=f(x)\cdot f(y)$.
Observe that this test accepts all linear functions with probability 1.
Blum, Luby and Rubinfeld proved that any function $f$ passing this test with high probability ($1-\delta$, for some small $\delta>0$), must agree with some linear function $\chi_S$ on most (at least $1-O(\delta)$ fraction of) points in $\set{0,1}^n$.
This result, with the acceptance/agreement probability close to 1, is known as the 99\%-regime of the test.

It was shown later~\cite{BCHKS96, KLX10} that the above result extends to the 1\% regime as well; more precisely, for every $\delta\in [0,1]$, and $f:\set{0,1}^n\to \set{-1,1}$ such that 
\[\E_{x,y\sim\set{0,1}^n}\sqbrac{f(x)\cdot f(y)\cdot f(x\oplus y)} \geq \delta,\]
there exists $S\subseteq [n]$ such that $\E_{x\sim \set{0,1}^n}\sqbrac{f(x)\cdot\chi_S(x)} \geq \delta$. 

The above test is of fundamental importance in theoretical computer science, and has several applications; for example, it is one of the ingredients in the proof of the celebrated PCP theorem~\cite{FGLSS96, AS98, ALMSS98}.
Furthermore, the analysis of the BLR test by Bellare et al.~\cite{BCHKS96} is one of the early uses of Fourier analysis over the boolean hypercube, an area which now plays a crucial role in many diverse subfields of mathematics and computer science, like complexity theory, harness of approximation, learning theory, coding theory, social choice theory, etc.~\cite{Don14}.

In this work, we are interested in the problem of linearity testing over the $p$-biased hypercube.
For $p\in (0,1)$, we denote by $\mu_p$ the $p$-biased distribution on $\set{0,1}$, which assigns probability $p$ to 1, and $1-p$ to 0.
The $p$-biased hypercube refers the set $\set{0,1}^n$, with the $n$-fold product measure $\mu_p^{\otimes n}$.
Linearity testing, in this $p$-biased setting, asks to distinguish between linear functions, and functions which are far (with respect to the $p$-biased measure) from being linear.

The 99\% regime of this problem is well-understood~\cite{KS09, DFH19}, and a simple 4-query test works in this case (see Example~\ref{eg:dfh19_test} below).
The question for the 1\% regime turns out to be significantly more challenging for any $p\not= 1/2$, and was wide open until a recent work of Bhangale, Khot and Minzer~\cite{BKM23b} made significant progress.
In particular, for every $p\in \brac{\frac{1}{3}, \frac{2}{3}}$, they give a 4-query test that works in the 1\% regime.

Building upon the work of Bhangale, Khot and Minzer, we consider a very general class of tests, where, very roughly, some $k$ queries $x_1,\dots, x_k \in \set{0,1}^n$, satisfying $\sum_{i\in [k]}x_i = 0 \modt$ are chosen, and the test accepts $f:\set{0,1}^n\to \set{-1,1}$ if $\prod_{i\in [k]} f(x_i) = 1$.
We shall require the following definitions:

\begin{definition}\label{defn:distr_class} (Class of Distributions)
	For $k\in \N,\ p\in (0,1)$, we define $\mc D(p,k)$ to be the class of all distributions $\nu$ on $\set{0,1}^k$ having $\mu_p$ as the marginal distribution on each coordinate $i\in [k]$, and such that $\supp(\nu)\subseteq \set{x\in \set{0,1}^k: \sum_{i=1}^kx_i = 0 \modt}$. 
	We say that such a distribution $\nu$ has \emph{full even-weight support}, if the above inclusion is an equality.
		
	For a distribution $\nu \in \mc D(p,k)$, we say that $i\in [k]$ is a \emph{pairwise independent coordinate}, if for each $j\in [k], j\not= i$, it holds that $\E_{X\sim \nu}\sqbrac{X_i\cdot X_j} = p^2$.
	We say that $\nu$ is \emph{pairwise independent}, if all its coordinates are pairwise independent.
\end{definition} 

\begin{definition}\label{defn:lin_test} (Class of Linearity Tests)
	For a distribution $\nu \in \mc D(p,k)$, we define a corresponding linearity test, denoted by $\Lin(\nu)$, as follows.
	Given query access to a function $f:\set{0,1}^n\to \set{-1,1}$:
	Sample\footnote{here, by $x = (x_1,\dots, x_k) \sim \nu^{\otimes n}$, we mean that for each $j\in [n]$, sample $(x_1\up{j},\dots, x_k\up{j})\sim \nu$ independently (also see Section~\ref{sec:prelims} for notation).} $x = (x_1,\dots, x_k) \sim \nu^{\otimes n}$, and accept if and only if $f(x_1)\cdot f(x_2)\cdots f(x_k) =1$.
\end{definition}

Note that every linear function passes such a test with probability 1\footnote{When $k$ is even, affine functions of the form $\pm \chi_S$ also pass the test with probability 1. In this work, we shall ignore the distinction between these functions and linear functions.}.
More strongly, each query in $\nu^{\otimes n}$ having marginal distribution $\mu_p^{\otimes n}$ ensures that functions that are close to linear (with respect to $p$-biased measure) are also accepted with high probability; in the property testing literature, such tests are called tolerant.
Furthermore, this is a very general class of linearity tests, containing many of the mentioned previously tests, as demonstrated by the following examples:

\begin{example}
	The BLR test uses $\nu$ to be uniform over ${\set{x\in \set{0,1}^3 : x_1+x_2+x_3 = 0 \modt}}$.
\end{example}
\begin{example}\label{eg:dfh19_test}
	The 4-query $p$-biased test of~\cite{DFH19} (for the 99\% regime) uses a distribution $\nu$ over $\set{0,1}^4$ of the following form:
	With probability $p_0$, set all coordinates to 0; with probability $p_1$, set all coordinates to 1; and with probability $1-p_0-p_1$, sample uniformly from the set $\set{x\in \set{0,1}^4 : x_1+x_2+x_3+x_4 = 0 \modt}$.
	Note that each coordinate has bias $p_1 + \frac{1}{2}\brac{1-p_0-p_1}$, and $p_0,p_1$ are chosen so that this equals $p$.
\end{example}

In this work, we analyze the precise conditions under which tests in Defintion~\ref{defn:lin_test} work for linearity testing, in the 1\% regime.
Our main result (proven in Section~\ref{sec:putting_together}) is the following:
\begin{theorem}\label{thm:intro_querybias_main_thm}
	Let $p\in (0,1)$.
	\begin{enumerate}
		\item For every integer $k > 1 + \frac{1}{\min\set{p,1-p}}$, there exists a distribution $\nu\in \mc D(p,k)$, such that the test $\Lin(\nu)$ is a $k$-query linearity test over the $p$-biased hypercube, for the 1\% regime.
		
		That is, for every $\eps>0$, there exists a $\delta>0$, such that for every large $n\in \N$, and every function $f:\set{0,1}^n\to [-1,1]$ satisfying\[ \abs{\E_{(X_1,\dots,X_k)\sim \nu^{\otimes n}} \sqbrac{\prod_{i\in [k]} f(X_i)} }\geq \eps,\]
		there exists a set $S\subseteq [n]$, such that $ \abs {\E_{X\sim \mu_p^{\otimes n}}\sqbrac{f(X)\cdot \chi_S(X)}} \geq \delta.$
		
		\item The above point also holds for all integers $k\geq 3$ with $p=\frac{1}{k-1}$, and for all even integers $k\geq 4$ with $p = 1-\frac{1}{k-1}$.
		
		\item Conversely, for every positive integer $k < 1+\frac{1}{\min\set{p,1-p}}$, and every distribution $\nu\in \mc D(p,k)$, the test $\Lin(\nu)$ fails in the 1\% regime.
		
		That is, there exists a constant $\alpha>0$, such that for every large $n\in \N$, there exists a function $f:\set{0,1}^n\to \set{-1,1}$ satisfying \[ \abs{\E_{(X_1,\dots,X_k)\sim \nu^{\otimes n}} \sqbrac{\prod_{i\in [k]} f(X_i)} }\geq \alpha,\]
		and such that for every $S\subseteq [n]$, it holds that $ \abs {\E_{X\sim \mu_p^{\otimes n}}\sqbrac{f(X)\cdot \chi_S(X)}} \leq o_n(1)$.
	\end{enumerate}
\end{theorem}

\begin{remark}\label{remark:corner_case_intro}
	Note that the above theorem does not discuss the case when $k\geq 5$ is an odd integer, and $p = 1-\frac{1}{k-1}$.
	This case is very interesting and is discussed in more detail in Section~\ref{sec:corner_case}.
	Informally speaking, the test corresponding to the ``natural" distribution $\nu \in \mc D(p,k)$, in this case, ensures correlation with a character of $\Z/(k-1)\Z$, and not a linear function $\chi_S$ (that is, a character of $\Z/2\Z$).
	In Section~\ref{sec:corner_case}, we also present an alternative test to get around this.
\end{remark}

Next, we shall describe the main technical results we prove along the way to prove Theorem~\ref{thm:intro_querybias_main_thm}.
We start by stating (a generalized version of) the main linearity testing result of Bhangale, Khot and Minzer~\cite{BKM23b}:

\begin{theorem}\label{thm:bkm23}
	(General version proved later as Theorem~\ref{thm:bkm23_in_section})
	Let $k\geq 3$ be a positive integer, and let $p\in (0,1),\ \epsilon \in (0,1]$ be constants, and let $\nu \in \mc D(p,k)$ be a distribution with full even-weight support (see Definition~\ref{defn:distr_class}).
	Then, there exists constants $\delta>0,\ d\in \N$ (possibly depending on $k, p, \epsilon, \nu$), such that for every large enough $n\in \N$, the following is true:
	
	Let $f:\set{0,1}^n\to[-1,1]$ be a function such that \[ \abs{\E_{(X_1,\dots,X_k)\sim \nu^{\otimes n}} \sqbrac{\prod_{i=1}^k f(X_i)} }\geq \eps.\]
	Then, there exists a set $S\subseteq [n]$, and a polynomial $g:\set{0,1}^n\to \R$ of degree at most $d$ and with 2-norm $\E_{X\sim \mu_p^{\otimes n} }\sqbrac{g(X)^2}\leq 1$, such that
	\[ \abs {\E_{X\sim \mu_p^{\otimes n}}\sqbrac{f(X)\cdot \chi_S(X)\cdot g(X)}} \geq \delta .\]
	
	Moreover, if the distribution $\nu$ has some pairwise independent coordinate, then we may assume $g\equiv 1$; that is, $f$ correlates with a linear function $\chi_S$.
\end{theorem}

We remark that Bhangale, Khot and Minzer only consider the case $k=4$, and only show $g\equiv 1$ in the case that all coordinates of $\nu$ are pairwise independent.
However, their proofs extend to the more general setting of Theorem~\ref{thm:bkm23}; we give an outline of this proof in Section~\ref{sec:bkm_sketch}.
Furthermore, we note that we are able to analyze the linearity test for a class of distributions which is much larger than the class of full even-weight support distributions; these distributions, in some sense, \emph{contain the BLR test}, and are formally defined in Section~\ref{sec:bkm_sketch}.

In the above work, the authors ask whether the conclusion $g\equiv 1$ can be obtained without assumption that $\nu$ has a pairwise independent coordinate.
We show this is not possible, and in fact the assumption that $\nu$ has a pairwise independent coordinate is necessary.

\begin{theorem}\label{thm:intro_main} (Restated and proved later as Theorem~\ref{thm:counter_eg})
	Let $k\in \N,\ p\in (0,1)$, and let $\nu \in \mc D(p,k)$ be a distribution having no pairwise independent coordinate (see Definition~\ref{defn:distr_class}).
	
	Then, there exists a constant $\alpha>0$, such that for every large enough $n\in \N$, there exists a function $f:\set{0,1}^n \to [-1,1]$ such that 
	\begin{enumerate}
		\item $\abs{\E_{X\sim \nu^{\otimes n}} \sqbrac{\prod_{i=1}^k f(X_i)} }\geq \alpha.$
		\item For every $S\subseteq [n]$, it holds that $ \abs{\E_{X\sim \mu_p^{\otimes n}}\sqbrac{f(X)\cdot \chi_S(X)}} \leq o_n(1)$.
	\end{enumerate}
	
	Moreover, if the distribution $\nu$ is such that $\eta:= \max_{i,j\in [k], i\not= j} \Pr_{X\sim \nu}\sqbrac{X_i=X_j} < 1$ (that is, no two coordinates are almost surely equal), the above holds for a function $f$ with range $\set{-1,1}$.
\end{theorem}

\begin{remark}
\hfill
\begin{enumerate}
	\item The assumption $\eta < 1$ in the second part of the Theorem~\ref{thm:intro_main} is necessary.
	For example, if the $i$\textsuperscript{th} and $j$\textsuperscript{th} coordinates of $\nu$ are equal, then, for functions $f$ with range $\set{-1,1}$, the terms $f(X_i)$ and $f(X_j)$ cancel out in the product $\E_{X\sim \nu^{\otimes n}} \sqbrac{\prod_{i=1}^k f(X_i)}$.
	In particular, the test is equivalent to the $(k-2)$-query test with coordinates $i,j$ removed from $\nu$, and this new distribution may possibly satisfy the conditions of Theorem~\ref{thm:bkm23}.
	\item The function $f$ we construct in Theorem~\ref{thm:intro_main}  does not correlate well with any linear function, although, as possibly required by Theorem~\ref{thm:bkm23}, it does correlate well with some constant degree function.
	\item The above theorem, answers in the \emph{negative} a question of~\cite{BKM23b}, who ask if \[\abs{\E_{(X,Y,Z,W)\sim \nu^{\otimes n}}\sqbrac{g_1(X)\cdot g_2(Y)\cdot g_3(Z)\cdot g_4(W)}} = o_n(1)\] for distributions $\nu\in \mc D(p,4)$ with full even-weight support, and $g_1,\dots,g_4:\set{0,1}^n\to \R$ bounded, noise stable, and resilient functions.
	\item It is an easy check that the distribution $\nu$ from Example~\ref{eg:dfh19_test} cannot have a pairwise independent coordinate, unless $p=1/2$. This shows that for $p\not=\frac{1}{2}$, simple tests that work in the 99\% regime fail to work in the 1\% regime.
	\item Recall that every $\nu\in\mc D(p,k)$ satisfies $\sum_i X_i=0 \modt$ almost surely, for $X\sim \nu$. We never use this in the proof of the above theorem, and the conclusion holds without it.
\end{enumerate}
\end{remark}

Very roughly speaking, in the proof of the above theorem we first construct a counter-example function in Gaussian space which ``passes the test" with decent probability, while having zero expectation; this function is then converted to a boolean function using the Central Limit Theorem and a rounding procedure.
Along the way, we prove a simple characterization for a random vector to have an independent coordinate, which we believe to be of independent interest, and is stated as follows:

\begin{proposition} (Restated formally and proved later as Proposition~\ref{prop:gaussian_counter_eg})
	Let $X=(X_1,\dots,X_k)$ be a $k$-dimensional multivariate Gaussian random vector, such that for each $i\in [k]$, the marginal is $X_i\sim \mc N(0,1)$.
	Then, the following are equivalent:
	\begin{enumerate}
		\item For every ``nice" function $f:\R\to \R$ satisfying $\E_{Z\sim \mc N(0,1)}\sqbrac{f(Z)}=0$, it holds that $\E\sqbrac{f(X_1)\cdot f(X_2)\cdots f(X_k)}=0$.
		\item There exists $i\in [k]$ such that $X_i$ is independent of $(X_1,\dots,X_{i-1}, X_{i+1}, \dots, X_k)$.
	\end{enumerate}
\end{proposition}

Finally, to use the above theorems (Theorem~\ref{thm:bkm23} and Theorem~\ref{thm:intro_main}), we analyze the tradeoff between the number of queries $k$ and the bias $p$, such that a distribution $\nu\in \mc D(p,k)$ with some pairwise independent coordinate exists.
In particular, we prove the following (restated and proved later as Proposition~\ref{prop:query_bias_lb} and Proposition~\ref{prop:query_bias_ub}): 

\begin{proposition}\label{prop:intro_bias_query_tradeoff}
	Let $k\in \N,\ p\in (0,1)$.
	Then, there exists a distribution $\nu \in \mc D(p,k)$ with some pairwise independent coordinate if and only if  $k \geq 1 + \frac{1}{\min\set{p,1-p}}$.
\end{proposition}

We note that the above generalizes the parameter setting for both the BLR test, corresponding to $p=\frac{1}{2},\ k=3$, and the case of $p\in \brac{\frac{1}{3},\frac{2}{3}},\ k=4$ considered in~\cite{BKM23b}.

\subsection{Related work}

The problem of linearity testing has been extensively studied, starting with the work of Blum, Luby and Rubinfeld~\cite{BLR93}, who gave a test for the uniform distribution, in the 99\% regime.
The analysis of their test was later extended to the 1\% regime~\cite{BCHKS96, KLX10}.
Tests for linearity have been also been studied in the low-randomness regime, and in the setting of non-abelian groups~\cite{BSVW03,BCLR08,SW06}.

For the $p$-biased case, in the 99\% regime, Halevy and Kushilevitz~\cite{HK07} gave a 3-query linearity test, that only uses random samples from the $p$-biased distribution!
However, the test is not tolerant, makes queries that are not distributed according to $\mu_p^{\otimes n}$, and hence may reject functions that are very close to linear (with respect to the $p$-biased measure).
Tolerant testers were analyzed later~\cite{KS09, DFH19}.
More strongly, the work of Dinur, Filmus and Harsha~\cite{DFH19} gives $2^d$-query tolerant tester for $p$-biased testing of degree $d$ functions over $\F_2$, a problem which has been well studied over the uniform distribution~\cite{AKKLR05, BKSSZ10}.

As a part of their work on approximability of satisfiable constraint satisfaction problems~\cite{BKM22, BKM23a, BKM23b, BKM24a, BKM24b}, Bhangale, Khot and Minzer study the $p$-biased version of linearity testing, in the 1\% regime.
As mentioned before, they give a 4-query test for $p\in \brac{\frac{1}{3},\frac{2}{3}}$. 

David, Dinur, Goldberg, Kindler and Shinkar~\cite{DDGKS17} study linearity testing on the $k$-slice (vectors of hamming-weight $k$), denoted by $L_{k,n}$, of the $n$-dimensional boolean hypercube, for even integers $k$. 
They show that if $f:\set{0,1}^n\to \set{-1,1}$ is such that $f(x\oplus y) = f(x)f(y)$ with probability $1-\epsilon$ over $x,y,x\oplus y$ (conditioned on all lying in $L_{k,n}$), then $f$ agrees with a linear function on $1-\delta$ fraction of $L_{k,n}$, where $\delta = \delta(\epsilon)\to 0$ as $\epsilon \to 0$.
In a recent work, Kalai, Lifshitz, Minzer and Ziegler~\cite{KLMZ24} prove a similar result for the $n/2$-slice, in the 1\% regime.

\subsection{Organization of the paper}

We start by presenting some preliminaries in Section~\ref{sec:prelims}.
In Section~\ref{sec:gaussian_variant}, we prove a variant of Theorem~\ref{thm:intro_main} over the Gaussian distribution, which then is used in Section~\ref{sec:lin_test_failure} to prove Theorem~\ref{thm:intro_main}.
In Section~\ref{sec:query_bias}, we analyze the tradeoff between the bias $p$ and the number of queries $k$, for the existence of a valid linearity test.
Combining all results, we prove Theorem~\ref{thm:intro_querybias_main_thm} in Section~\ref{sec:putting_together}.
In Section~\ref{sec:bkm_sketch}, we outline of the proof of Theorem~\ref{thm:bkm23}.

\section{Preliminaries}\label{sec:prelims}

We use $\exp$ to denote the exponential function, given by $\exp(x) = e^x$ for $x\in \R$.

Let $\N = \set{1,2,\dots}$ be the set of natural numbers. For each $n\in \N$, we use $[n]$ to denote the set $\set{1,2,\dots,n}$.
For non-negative functions $f,g:\N\to \R$, we say that $f(n) = o_n(g(n))$ if $\lim_{n\to \infty} \frac{f(n)}{g(n)}=0$.

For a probability distribution $\nu$ on $\mc X$, we use $\supp(\nu)$ to denote its support.
For $n\in \N$, we use $\nu^{\otimes n}$ to denote the $n$-fold product distribution on $\mc X^n$.
In particular, we shall be interested in the case when $\mc X \subseteq  \R^k$ for some $k\in \N$.
In this case, for vectors $x \in \R^{kn}$, we shall use subscripts for indices in $[k]$ and superscripts for indices in $[n]$; that is, for each $i\in [k], j\in [n]$, we use $x_i\up{j}$ to denote the $(i,j)^\textsuperscript{th}$ coordinate of $x$.
Further, for each $i\in [k]$, we use $x_i$ to denote the vector $\brac{x_i\up{1},\dots,x_i\up{n}}\in \R^n$, and similarly for each $j\in [n]$, we use $x\up{j}$ to denote the vector $\brac{x_1\up{j},\dots,x_k\up{j}}\in \R^k$.

For $k\in \N$, let $S_k$ denote the group of all permutations on $[k]$.
	For each $\pi\in S_k,\ x\in \R^k$, we use $x_\pi$ to denote $\brac{x_{\pi(1)},\dots,x_{\pi(k)}} \in \R^k$.
With this notation, we define the symmetrization of functions over $\R^k$:
\begin{definition}\label{defn:sym_op}
	For any function $f:\R^k\to \R$, we define its symmetrization as the function $\Sym(f):\R^k\to \R$, given by $\Sym(f)(x) = \sum_{\pi\in S_k}f(x_\pi)$.
\end{definition}

We shall use the following facts from probability theory:

\begin{fact}\label{fact:chebyshev}(Chebyshev's Inequality; see~\cite{Dur19} for reference)
	Let $X$ be a random variable such that $\E\sqbrac{X^2}<\infty$.
	Then, for any any $a>0$,
	\[\Pr\sqbrac{\abs{X-\E\sqbrac{X}} \geq a} \leq \frac{\Var\sqbrac{X}}{a^2}.\]
\end{fact}

\begin{fact}\label{fact:hoeffding}(Hoeffding's Inequality~\cite{Hoe63})
	Let $X_1,\dots, X_n$ be independent random variables such that $a_i\leq X_i \leq b_i$ almost surely, and let $S = \sum_{i=1}^n X_i$.
	Then, for all $t>0$,
	\[\Pr\sqbrac{\abs{S - \E\sqbrac{S}} \geq  t} \leq 2\cdot \exp\brac{-\frac{2t^2}{\sum_{i=1}^n \brac{b_i-a_i}^2}}. \]
\end{fact}

\begin{theorem}\label{thm:multi_clt} (Multivariate Central Limit Theorem; see~\cite{Dur19} for reference)
Let $X\up{1},X\up{2},\dots$ be $\R^k$-valued i.i.d. random vectors, with mean zero, and finite a covariance matrix $\Sigma \in \R^{k\times k}$ given by $\Sigma_{i,j} = \E\sqbrac{X\up{1}_i\cdot X\up{1}_j}$.
If $S_n = \frac{1}{\sqrt n} \sum_{i=1}^n X\up{i}$, then, $S_n \xrightarrow{\mc D} Z$ as $n\to \infty$, for $Z\sim  \mc N(0, \Sigma)$.
That is, for every bounded continuous function $H:\R^k\to \R$, \[\lim_{n\to \infty} \E\sqbrac{H(S_n)} = \E_{Z\sim \mc N(0,\Sigma)}\sqbrac{H(Z)}.\]
	
\end{theorem}

We shall also use the following fact about zeros of polynomials:
\begin{lemma}\label{lemma:dim_var}
	Let $p_1,\dots,p_r:\R^k\to \R$ be non-zero polynomials.
	Then, there exists $y\in \R^k$ such that for each permutation $\pi\in S_k$, and each $i\in [r]$, it holds that $p_i(y_\pi) \not= 0$.
\end{lemma}
\begin{proof}[Proof Sketch]
	The zero-set of any non-zero polynomial has measure zero, with respect to the Lebesgue measure on $\R^k$.
	Hence, by sub-additivity, the set of points $y\in \R^k$ violating the statement of the lemma is of measure zero as well.
\end{proof}

Next, we give some basic results about the probabilist's Hermite polynomials. The reader is referred to Chapter 11 in~\cite{Don14} for more details. 

\begin{definition}\label{defn:hermite_poly}
	The Hermite polynomials $(H_j)_{j\in \Z_{\geq 0}}$ are univariate polynomials, with $H_j$ a monic polynomial of degree $j$, satisfying the power series expression
	\[ \exp\brac{tx-\frac{t^2}{2}} = \sum_{j=0}^\infty \frac{1}{j!}\cdot H_j(x)\cdot t^j, \quad \text{for } t,x\in \R. \]
	Note that the series above is absolutely convergent, with $  \sum_{j=0}^\infty \frac{1}{j!}\cdot \abs{H_j(x)}\cdot \abs{t}^j \leq  \exp\brac{\abs{t}\cdot \abs{x}+\frac{t^2}{2}}$ for each $t,x\in \R$.
\end{definition}

\begin{lemma}\label{lemma:hermite_exp}
	Let $k\in \N$ and $s_1,s_2,\dots,s_k \in \Z_{\geq 0}$, and let $\Sigma\in \R^{k\times k}$ be a positive semi-definite matrix such that $\Sigma_{i,i}=1$ for each $i$.	
	For $V=\Sigma-I$, it holds that \[\E_{X\sim \mc N(0,\Sigma)}\sqbrac{H_{s_1}(X_1)\cdots H_{s_k}(X_k)} = \frac{s_1!\cdots s_k!}{d!\cdot 2^d}\cdot \sqbrac{\brac{t^{\top}V\ t}^d : t_1^{s_1}\cdots t_k^{s_k}} \]
	where $s_1+\dots+s_k = 2d$, and $\sqbrac{\brac{t^{\top}V\ t}^d: t_1^{s_1}\cdots t_k^{s_k}}$ denotes the coefficient of $t_1^{s_1}\cdots t_k^{s_k}$ in the polynomial $\brac{t^{\top}V\ t}^d$.
	Also, the above expectation is zero when $s_1+\dots+s_k$ is odd.
\end{lemma}
\begin{proof}
	Recall that the moment generating function of a multivariate Gaussian distribution is given by
	\[\E_{X\sim \mc N(0,\Sigma)}\sqbrac{\exp\brac{t_1X_1+\dots t_kX_k}} = \exp\brac{\frac{1}{2}\cdot t^{\top}\Sigma\ t},\]
	for each $t\in \R^k$.
	Multiplying the above by $\exp(-\frac{1}{2}\cdot t^\top t)$, and plugging in the power series in Definition~\ref{defn:hermite_poly}, we get for each $t\in \R^k$ that
	\begin{align*}
		\sum_{d=0}^{\infty} \frac{1}{d!\cdot 2^d}\cdot \brac{t^{\top}V\ t}^d
		&= \exp\brac{\frac{1}{2}\cdot t^{\top}V\ t} 
		\\&= \E_{X\sim \mc N(0,\Sigma)}\sqbrac{ \exp\brac{\brac{t_1X_1-\frac{t_1^2}{2}}+\dots +\brac{t_kX_k -\frac{t_k^2}{2}} } }
		\\&= \E_{X\sim \mc N(0, \Sigma)}\sqbrac{ \brac{\sum_{s_1=0}^\infty \frac{1}{s_1!}\cdot H_{s_1}(X_1)\cdot t_1^{s_1} }\cdots \brac{\sum_{s_k=0}^\infty \frac{1}{s_k!}\cdot H_{s_k}(X_k)\cdot t_k^{s_k}}}
		\\&= \sum_{s_1,\dots,s_k\geq 0}\frac{t_1^{s_1}\cdots t_k^{s_k}}{s_1!\cdots s_k!}\cdot \E_{X\sim \mc N(0, \Sigma)}\sqbrac{H_{s_1}(X_1)\cdots H_{s_k}(X_k)}.
	\end{align*}
	Note that since the power series in Definition~\ref{defn:hermite_poly} is absolutely convergent, all steps above of interchanging limits and expectations are valid by the dominated convergence theorem. 
	Finally, comparing coefficients, we have the desired result.
\end{proof}

\section{A Gaussian Variant}\label{sec:gaussian_variant}

The first step towards proving Theorem~\ref{thm:intro_main} is to prove a Gaussian variant, stated below:

\begin{proposition}\label{prop:gaussian_counter_eg}	
	Let $k\in \N$, and let $\Sigma\in \R^{k\times k}$ be a symmetric positive semi-definite matrix such that:
	\begin{enumerate}
		\item For each $i\in [k]$, it holds that $\Sigma_{i,i} = 1$.
		\item The matrix $V = \Sigma - I$ has no row/column as all zeros.
	\end{enumerate}
	
	Then, there exists a Lipschitz continuous function $f:\R\to [-1,1]$ such that:
	\[\E_{X\sim \mc N(0,1)} \sqbrac{f(X)} = 0, \quad\text{and}\quad \abs{\E_{X\sim \mc N(0, \Sigma)}\sqbrac{\prod_{i\in [k]} f(X_i)}} > 0 .\]
\end{proposition}

%%%%%%%%%%%%%%%%%%%%%%%%%%%%%%%%%%%%%%%%%%%%%%%%%%%%%%%%%%%%%%%%%%%%%%%%%%%%%%%%
\subsection{Symmetric Powers of Polynomials}

Before we prove the above proposition, we first prove a lemma about (symmetrization of) powers of multivariate polynomials.
We show that if a polynomial $q(x_1,\dots,x_k)$ depends on all the variables $x_1,\dots,x_k$, then some power $\Sym(q^d)$ (see Definition~\ref{defn:sym_op}) contains a monomial divisible by $x_1x_2\cdots x_k$.

\begin{lemma}\label{lemma:polynomial_all_var}
	Let $k\in \N$, and let $q:\R^k\to \R$ be a polynomial such that for each $i\in [k]$, the polynomial $\ell_i=\partial_iq$ is not identically zero.
	Then, there exists some $d\in \N$, and positive integers $s_1,\dots,s_k \in \N$ such that the coefficient of $x_1^{s_1}\cdot x_2^{s_2}\cdots x_k^{s_k}$ in the polynomial $\Sym(q^d)$ is non-zero.
\end{lemma}

We start by proving the following lemma about derivates of powers of $q$.
\begin{lemma}\label{lemma:pow_derivatives}
	Let $k\in \N$, and let $q:\R^k\to \R$ be a polynomial. For each $i\in [k]$, let $\ell_i=\partial_iq$.
	
	Then, for every $s=(s_1,\dots,s_k)\in \Z_{\geq 0}^k$ with $\abs{s}=\sum_{i\in [k]}s_i$, there exist  polynomials $p_0, \dots, p_{\abs{s}}$, with $p_{\abs{s}} = \prod_{i\in [k]}\ell_i^{s_i}$, such that for each $d\geq \abs{s}$, it holds that 
	\[ \partial_1^{s_1}\cdot\partial_2^{s_2}\cdots \partial_k^{s_k}\brac{q^d} = q^{d-\abs{s}}\cdot \brac{\sum_{i=0}^{\abs{s}} d^i\cdot p_i}.\]
\end{lemma}
\begin{proof}
	The proof is by induction on $\abs{s}$.
	For the base case, if $\abs{s}=0$, we have $s = (0,0,\dots,0)$, and $p_0 = 1$ satisfies the statement of the lemma.
	
	For the inductive step, consider any $s=(s_1,\dots,s_k)\in \Z_{\geq 0}^k$ with $\abs{s}=\sum_{i\in [k]}s_i > 0$.
	Without loss of generality, by symmetry, we can assume that $s_1>0$.
	By the inductive hypothesis applied to $(s_1-1,s_2\dots,s_k)$, we have the existence of polynomials $p_0, \dots, p_{\abs{s}-1}$, with $p_{\abs{s}-1} = \ell_1^{s_1-1}\cdot \prod_{i=2}^k\ell_i^{s_i}$, and such that for each $d\geq \abs{s}-1$, we have 
	\[ \partial_1^{s_1-1}\cdot\partial_2^{s_2}\cdots \partial_k^{s_k}\brac{q^d} = q^{d-\abs{s}+1}\cdot \brac{\sum_{i=0}^{\abs{s}-1} d^i\cdot p_i}.\]
	Now, if $d\geq \abs{s}$, differentiating the above with respect to $x_1$, we get 
	\begin{align*}
		\partial_1^{s_1}\cdot\partial_2^{s_2}\cdots \partial_k^{s_k}\brac{q^d} 
		&= q^{d-\abs{s}}\cdot \brac{(d-\abs{s}+1)\cdot \ell_1\cdot \sum_{i=0}^{\abs{s}-1} d^i\cdot p_i} + q^{d-\abs{s}+1}\cdot \brac{\sum_{i=0}^{\abs{s}-1} d^i\cdot \partial_1(p_i)}
		\\&= q^{d-\abs{s}}\cdot \brac{\sum_{i=1}^{\abs{s}} d^{i}\cdot \ell_1\cdot p_{i-1} + \sum_{i=0}^{\abs{s}-1}d^i\cdot \brac{\brac{-\abs{s}+1}\cdot \ell_1\cdot p_i + q\cdot \partial_1 p_i} }
		\\&= q^{d-\abs{s}}\cdot \brac{\sum_{i=0}^{\abs{s}} d^{i}\cdot \tilde{p}_i },
	\end{align*}
	where the polynomials $\tilde{p}_1,\dots,\tilde{p}_{\abs{s}}$ do not depend on $d$, and are such that $\tilde{p}_{\abs{s}} = p_{\abs{s}-1}\cdot \ell_1  = \prod_{i\in [k]}\ell_i^{s_i}$, as desired.
\end{proof}

With the above lemma in hand, next we shall consider the symmetrization operation applied to derivatives of powers of $q$.

\begin{lemma}\label{lemma:pow_sec_der}
	Let $k\in \N$, and let $q:\R^k\to \R$ be a polynomial such that for each $i\in [k]$, the polynomial $\ell_i=\partial_iq$ is not identically zero.
	
	Then, for each large enough even integer $d\in \N$, the polynomial $\Sym\brac{\partial_1^2\cdot\partial_2^2\cdots \partial_k^2\brac{q^d}}$ is not identically zero.
\end{lemma}
\begin{proof}
	By applying Lemma~\ref{lemma:pow_derivatives} on $s = (2,2,\dots,2)$, we have the existence of polynomials $p_0, \dots, p_{2k}$, with $p_{2k} = \prod_{i\in [k]}\ell_i^2$, such that for each $d\geq 2k$, it holds that $\partial_1^{2}\cdot\partial_2^{2}\cdots \partial_k^{2}\brac{q^d} = q^{d-2k}\cdot \brac{\sum_{i=0}^{2k} d^i\cdot p_i}$.
	
	By Lemma~\ref{lemma:dim_var}, let $y\in \R^k$ be such that $y$ (and its permutations) don't lie in the zero set of any of the polynomials $\ell_1,\dots,\ell_k, q$.
	We define \[ A = \min_{\pi\in S_k}\sqbrac{\prod_{i\in [k]}\ell_i\brac{y_\pi}^2} > 0, \quad B = \max_{0\leq i\leq 2k-1,\ \pi\in S_k}\abs{p_i(y_\pi)}\geq 0. \]
	Then, for any even integer $d\geq \max\set{2k, \frac{4kB}{A}}$, it holds that
	\begin{align*}
		\Sym\brac{\partial_1^{2}\cdot\partial_2^{2}\cdots \partial_k^{2}\brac{q^d}}(y) 
		&=\sum_{\pi\in S_k} q\brac{y_\pi}^{d-2k}\cdot \brac{d^{2k}\cdot \prod_{i\in [k]}\ell_i\brac{y_{\pi}}^2 + \sum_{i=0}^{2k-1} d^i\cdot p_i\brac{y_\pi}}
		\\&\geq \sum_{\pi\in S_k} q\brac{y_\pi}^{d-2k}\cdot \brac{d^{2k}\cdot A - \sum_{i=0}^{2k-1}d^i\cdot B}
		\\&\geq  \brac{\sum_{\pi\in S_k} q\brac{y_\pi}^{d-2k}}\cdot \brac{d^{2k}\cdot A - 2k\cdot d^{2k-1}\cdot B}
		\\&\geq \brac{\sum_{\pi\in S_k} q\brac{y_\pi}^{d-2k}}\cdot \frac{d^{2k} A}{2} > 0.
	\end{align*}
	Hence, for even integers $d\geq \max\set{2k, \frac{4kB}{A}}$, the polynomial $\Sym\brac{\partial_1^{2}\cdot\partial_2^{2}\cdots \partial_k^{2}\brac{q^d}}$ is not identically zero.
\end{proof}

Finally, we prove the main lemma of this section.

\begin{proof}[Proof of Lemma~\ref{lemma:polynomial_all_var}]
	Let $k\in \N$, and let $q:\R^k\to \R$ be a polynomial such that for each $i\in [k]$, the polynomial $\ell_i=\partial_iq$ is not identically zero.
	It suffices to prove that for some $d\in \N$, the polynomial $\partial_1^{2}\cdot\partial_2^{2}\cdots \partial_k^{2}\brac{\Sym(q^d)}$ is not identically zero, since then the coefficient of some monomial divisible by $x_1^2\cdot x_2^2\cdots x_k^2$ is non-zero.
	
	For each polynomial $p:\R^k\to \R$, and each $\pi\in S_k$, we shall use $p_\pi$ to denote the polynomial given by $p_\pi(x) = p(x_\pi)$.
	Then, for all $s_1,\dots,s_k\in \Z_{\geq 0}$, we have that $\partial_1^{s_1}\cdot\partial_2^{s_2}\cdots \partial_k^{s_k} \brac{p_\pi} = \brac{\partial_{\pi^{-1}(1)}^{s_1}\cdot\partial_{\pi^{-1}(2)}^{s_2}\cdots \partial_{\pi^{-1}(k)}^{s_k} \brac{p}}_\pi$.
	
	By the above, we have that for each $d\in \N$,
	\begin{align*}
		\partial_1^{2}\cdot\partial_2^{2}\cdots \partial_k^{2}\brac{\Sym(q^d)} 
		&=  \partial_1^{2}\cdot\partial_2^{2}\cdots \partial_k^{2}\brac{\sum_{\pi\in S_k}q_\pi^d}
		\\&= \sum_{\pi\in S_k} \partial_1^{2}\cdot\partial_2^{2}\cdots \partial_k^{2}\brac{q_\pi^d}
		\\&= \sum_{\pi\in S_k} \brac{\partial_{\pi^{-1}(1)}^{2}\cdot\partial_{\pi^{-1}(2)}^{2}\cdots \partial_{\pi^{-1}(k)}^{2} \brac{q^d}}_\pi
		\\&= \sum_{\pi\in S_k} \brac{\partial_{1}^{2}\cdot\partial_{2}^{2}\cdots \partial_{k}^{2} \brac{q^d}}_\pi
		\\&= \Sym\brac{\partial_1^{2}\cdot\partial_2^{2}\cdots \partial_k^{2}\brac{q^d}}.
	\end{align*}	
	Now, the result follows from Lemma~\ref{lemma:pow_sec_der}.
\end{proof}

%%%%%%%%%%%%%%%%%%%%%%%%%%%%%%%%%%%%%%%%%%%%%%%%%%%%%%%%%%%%%%%%%%%%%%%%%%%%%%%%
\subsection{Proving the Gaussian Variant}

We start by proving a slight variant of Proposition~\ref{prop:gaussian_counter_eg}, where we allow $f$ to be an arbitrary (possibly unbounded) polynomial.

\begin{lemma}\label{lemma:hermite_counter_eg}
	Let $k\in \N$, and let $\Sigma\in \R^{k\times k}$ be a symmetric positive semi-definite matrix such that:
	\begin{enumerate}
		\item For each $i\in [k]$, it holds that $\Sigma_{i,i} = 1$.
		\item The matrix $V = \Sigma - I$ has no row/column as all zeros.
	\end{enumerate}
	Then, there exists a polynomial $f:\R\to \R$ such that $\E_{X\sim \mc N(0,1)} \sqbrac{f(X)} = 0$, and \[ \abs{\E_{X\sim \mc N(0, \Sigma)}\sqbrac{\prod_{i\in [k]} f(X_i)}} > 0 .\]
\end{lemma}

\begin{proof}
	For $s = (s_1,\dots,s_k)\in \N^k$ and $\alpha = (\alpha_1,\dots,\alpha_k)\in \R^{k}$, let $f_{s,\alpha}:\R\to \R$ be the polynomial defined by $f_{s,\alpha}(x) = \alpha_1 H_{s_1}(x)+\cdots + \alpha_k H_{s_k}(x)$, where the polynomials $H_{s_i}$ are Hermite polynomials (see Definition~\ref{defn:hermite_poly}).	
	Observe that since $s_1,\dots,s_k \geq 1$, this polynomial satisfies $\E_{X\sim \mc N(0,1)} \sqbrac{f(X)} = 0$.

	Suppose, for the sake of contradiction, that for every $s\in \N^k,\ \alpha\in \R^k$, it holds that 
	\[\E_{X\sim \mc N(0, \Sigma)}\sqbrac{\prod_{i\in [k]} f_{s,\alpha}(X_i)} = \E_{X\sim \mc N(0, \Sigma)}\sqbrac{\prod_{i\in [k]} \sum_{j\in [k]}\alpha_jH_{s_j}(X_i) } =  0.\]
	Observe that for every $s\in \N^k$, the above expression can be written as a multivariate polynomial in $\alpha_1,\dots,\alpha_k$.
	If the polynomial vanishes for all $\alpha\in \R^k$, the coefficient of $\alpha_1\cdot \alpha_2\cdots \alpha_k$ must be zero; that is,
	\[ \sum_{\pi\in S_k}\E_{X\sim \mc N(0, \Sigma)}\sqbrac{\prod_{i\in [k]} H_{s_{\pi(i)}}(X_i)} = 0.\]
	Now, applying Lemma~\ref{lemma:hermite_exp}, we get that for each $d\in \N$, and each $s_1,\dots,s_k \geq 1$ with $s_1+\dots+s_k=2d$,
	\[\sum_{\pi\in S_k}\sqbrac{\brac{t^{\top}V\ t}^d : t_1^{s_{\pi(1)}}\cdots t_k^{s_{\pi(k)}}}
		= \sum_{\pi\in S_k}\sqbrac{\brac{t_\pi^{\top}V\ t_\pi}^d : t_{1}^{s_k}\cdots t_{k}^{s_k}}
		= \sqbrac{ \Sym\brac{\brac{t^{\top}V\ t}^d} : t_{1}^{s_k}\cdots t_{k}^{s_k}} = 0. \]
	Note that the assumption that $V$ has no zero row/column implies that for every $i\in [k]$, the polynomial $\partial_i \brac{t^{\top}V\ t}$ is not identically zero.
	By Lemma~\ref{lemma:polynomial_all_var}, this is a contradiction.
\end{proof}

With the above, we now prove Proposition~\ref{prop:gaussian_counter_eg} via a standard truncation argument.

\begin{proof}[Proof of Proposition~\ref{prop:gaussian_counter_eg}]
By Lemma~\ref{lemma:hermite_counter_eg}, we know that there exists a polynomial $f:\R\to \R$ such that $\E_{X\sim \mc N(0,1)} \sqbrac{f(X)} = 0$, and $\abs{\E_{X\sim \mc N(0, \Sigma)}\sqbrac{\prod_{i\in [k]} f(X_i)}} > 0$.

For each integer $M\in \N$, we define the truncated function $f_M:\R\to [-M,M]$ by 
\[f_M(x) = f(x)\cdot \ind_{\abs{f(x)}\leq M} + M\cdot \ind_{f(x) > M} - M\cdot \ind_{f(x) < -M}.\]
Also, let $g_M:\R\to [-2M,2M]$, be given by $g_M(x) = f_M(x) - \E_{X\sim \mc N(0,1)}\sqbrac{f_M(X)}$.
Observe that
\begin{enumerate}
	\item For every $M$, it holds that $\E_{X\sim \mc N(0,1)}\sqbrac{g_M(X)} = 0$.
	\item For every $M$, the function $g_M$ is bounded and Lipschitz continuous.
	\item For every $x\in \R$, $f_M(x)\to f(x)$ as $M\to \infty$. Further, since $\abs{f_M(x)} \leq \abs{f(x)}$ for each $x\in \R, M\in \N$, by the dominated convergence theorem, we have  $\E_{X\sim \mc N(0,1)}\sqbrac{f_M(x)}\to \E_{X\sim \mc N(0,1)}\sqbrac{f(x)}=0$ as $M\to \infty$.
	This implies that for each $x\in \R$, $g_M(x)\to f(x)$ as $M\to \infty$.

	Also, for each $x\in \R, M\in \N$, we have $\abs{g_M(x)} \leq \abs{f(x)} + \E_{X\sim \mc N(0,1)}\sqbrac{\abs{f(X)}}$. Hence, by the dominated convergence theorem, we have that $\E_{X\sim \mc N(0, \Sigma)}\sqbrac{\prod_{i\in [k]} g_M(X_i)} \to \E_{X\sim \mc N(0, \Sigma)}\sqbrac{\prod_{i\in [k]} f(X_i)} \not= 0$ as $M\to \infty$.
\end{enumerate}

By the above, for some large enough $M$, the function $\frac{1}{2M}\cdot g_M:\R\to [-1,1]$ satisfies the desired properties.
\end{proof}

%%%%%%%%%%%%%%%%%%%%%%%%%%%%%%%%%%%%%%%%%%%%%%%%%%%%%%%%%%%%%%%%%%%%%%%%%%%%%%%%
\section{Linearity Testing Requires Pairwise Independence}\label{sec:lin_test_failure}

In this section, we prove Theorem~\ref{thm:intro_main}, which is restated below.

\begin{theorem}\label{thm:counter_eg}
	Let $k\in \N,\ p\in (0,1)$, and let $\nu \in \mc D(p,k)$ be a distribution having no pairwise independent coordinate (see Definition~\ref{defn:distr_class}).
	Then, there exists a constant $\alpha>0$, such that for every large enough $n\in \N$, there exists a function $f:\set{0,1}^n \to [-1,1]$ such that 
	\begin{enumerate}
		\item $\abs{\E_{X\sim \nu^{\otimes n}} \sqbrac{\prod_{i=1}^k f(X_i)} }\geq \alpha.$
		\item For every $S\subseteq [n]$, it holds that $ \abs{\E_{X\sim \mu_p^{\otimes n}}\sqbrac{f(X)\chi_S(X)}} \leq o_n(1)$.
	\end{enumerate}
	
	Moreover, if the distribution $\nu$ is such that $\eta:= \max_{i,j\in [k], i\not= j} \Pr_{X\sim \nu}\sqbrac{X_i=X_j} < 1$ (that is, no two coordinates are almost surely equal), the above holds for a function $f$ with range $\set{-1,1}$.
\end{theorem}

The remainder of this section is devoted to the proof of Theorem~\ref{thm:counter_eg}.
In Section~\ref{sec:counter_eg_1}, we prove the first part of the theorem, dealing with functions with range $[-1,1]$.
Then, in Section~\ref{sec:counter_eg_2}, we show how to round to functions with range $\set{-1,1}$.

%%%%%%%%%%%%%%%%%%%%%%%%%%%%%%%%%%%%%%%%%%%%%%%%%%%%%%%%%%%%%%%%%%%%%%

\subsection{Function with Range $[-1,1]$}\label{sec:counter_eg_1}

Let $k\in \N,\ p\in (0,1),$ and let $\nu \in \mc D(p,k)$ be a distribution having no pairwise independent coordinate.
Let $\Sigma \in \R^{k\times k}$ be the (normalized) covariance matrix corresponding to the distribution $\nu$, given by, $\Sigma_{i,j} = \E_{X\sim \nu}\sqbrac{\frac{ \brac{X_i-p}\cdot \brac{X_j-p}}{p-p^2}}$.
Observe that the matrix $\Sigma$ satisfies the conditions of Proposition~\ref{prop:gaussian_counter_eg}, and hence there exists a function $h:\R\to [-1,1]$ such that
\begin{enumerate}
	\item $\E_{Z\sim \mc N(0,1)}\sqbrac{h(Z)} = 0$
	\item The function $H:\R^k\to [-1,1]$ given by $H(x) = \prod_{i\in [k]}h\brac{x_i}$ is such that \[\alpha :=  \frac{1}{2}\cdot \abs{\E_{Z\sim \mc N(0, \Sigma)} \sqbrac{H(Z)}} > 0 .\]
	\item The function $h$ is $K$-Lipschitz for some $K>0$; in particular, both $h$ and $H$ are bounded continuous functions.
\end{enumerate}

Consider any large $n\in \N$.
We define $f:\set{0,1}^n \to [-1,1]$ by 
\[f(x) = h\brac{\frac{1}{\sqrt{n}}\cdot \sum_{j=1}^n \frac{x\up{j}-p}{\sqrt{p-p^2}}},\]
The function $f$ satisfies the two properties in the theorem statement, as follows:

\begin{itemize}
	\item Let $X \sim \nu^{\otimes n}$, and let $Y = (Y_1,\dots, Y_k)$ be a $\set{0,1}^k$-valued random vector, defined as $Y_i=\frac{1}{\sqrt{n}}\cdot \sum_{j=1}^n \frac{X_i\up{j}-p}{\sqrt{p-p^2}}$.
	 
	Let $F:\set{0,1}^{kn}\to [-1,1]$ be given by $F(x) = \prod_{i\in [k]}f\brac{x_i}$.
	Since $H$ is continuous and bounded, we have by the Multivariate CLT (Theorem~\ref{thm:multi_clt}) that
	\[ \abs{\ \E\sqbrac{F(X)} - \E_{Z \sim \mc N(0, \Sigma)} \sqbrac{H(Z)}\ } = \abs{\ \E\sqbrac{{H(Y)}} - \E_{Z \sim \mc N(0, \Sigma)} \sqbrac{H(Z)}\ } \leq o_n(1).\]
	Hence, for large $n$, we get $\abs{\E_{X\sim \nu^{\otimes n}} \sqbrac{\prod_{i=1}^k f(X_i)} }\geq 2\alpha-o_n(1)\geq \alpha$, as desired.

	\item Consider any subset $S\subseteq [n]$, and let $T\subseteq S$ be any subset of size $\abs{T}=\min\set{\lfloor n^{1/4}\rfloor, \abs{S}}$. Let $\tilde{f}:\set{0,1}^n\to [-1,1]$ be defined by $\tilde{f}(X) = h\brac{\frac{1}{\sqrt{n-\abs{T}}}\cdot \sum_{j\in [n]\setminus T} \frac{x\up{j}-p}{\sqrt{p-p^2}} }$; note that this function only depends on the coordinates of $x$ outside the set $T$. Further, for each $x\in \set{0,1}^n$, by the Lipschitz bound on $h$, we get
	\begin{align*}
		\abs{f(x)-\tilde{f}(x)} &\leq K\cdot  \abs{\frac{1}{\sqrt{n}}\cdot \sum_{j=1}^n \frac{x\up{j}-p}{\sqrt{p-p^2}} - \frac{1}{\sqrt{n-\abs{T}}}\cdot \sum_{j\in [n]\setminus T} \frac{x\up{j}-p}{\sqrt{p-p^2}}}
		\\&\leq \frac{K}{\sqrt{p-p^2}}\cdot\brac{\frac{\abs{T}}{\sqrt{n}}+ \brac{n-\abs{T}}\cdot \abs{\frac{1}{\sqrt{n-\abs{T}}} - \frac{1}{\sqrt{n}}}}
		\\&\leq \frac{K}{\sqrt{p-p^2}}\cdot\brac{\frac{\abs{T}}{\sqrt{n}}+ \frac{n-\abs{T}}{\sqrt{n}}\cdot \frac{\abs{T}}{n}}
		\\&\leq \frac{K}{\sqrt{p-p^2}}\cdot\frac{2\abs{T}}{\sqrt{n}} = o_n(1),
	\end{align*}
	where we used that $(1-t)^{-1/2} \leq 1+t$ for each $t\in [0,1/2]$.
 
	Now, for $X\sim \mu_p^{\otimes n}$, we have
	\begin{align*}
		\abs{\ \E_X\sqbrac{f(X)\cdot \chi_S(X)}\ } &\leq \abs{\ \E_X\sqbrac{\tilde{f}(X)\cdot \chi_S(X)}\ } + o_n(1)
		\\&= \abs{\ \E_X\sqbrac{\tilde{f}(X)\cdot \chi_{S\setminus T}(X)}\cdot \E_X\sqbrac{ \chi_{T}(X)}\ } + o_n(1)
		\\&= \abs{\ \E_X\sqbrac{\tilde{f}(X)\cdot \chi_{S\setminus T}(X)}\ }\cdot \abs{1-2p}^{\abs{T}} + o_n(1).
	\end{align*}
	If $\abs{S} \geq \lfloor n^{1/4} \rfloor$, then $\abs{1-2p}^{\abs{T}} = o_n(1)$.
	Otherwise, we have that $S=T$, and by the Central Limit Theorem (see Theorem~\ref{thm:multi_clt}) , the first term in the above product equals
	\[ \abs{\ \E_X\sqbrac{\tilde{f}(X)}\ } = \abs{\ \E_X\sqbrac{\tilde{f}(X)} - \E_{Z\sim \mc N(0,1)}\sqbrac{h(Z)}\ } = o_n(1). \pushQED{\qed}\qedhere\popQED \]	
\end{itemize}

%%%%%%%%%%%%%%%%%%%%%%%%%%%%%%%%%%%%%%%%%%%%%%%%%%%%%%%%%%%%%%%%%%%%%%%%%%%%%%%%

\subsection{Rounding to a Function with Range $\set{-1,1}$}\label{sec:counter_eg_2}

Now, we shall prove the second part of Theorem~\ref{thm:counter_eg}.

Let $k\in \N,\ p\in (0,1),$ and let $\nu \in \mc D(p,k)$ be a distribution having no pairwise independent coordinate.
Further suppose that the distribution $\nu$ is such that \[\eta := \max_{i,j\in [k], i\not= j} \Pr_{X\sim \nu}\sqbrac{X_i=X_j} < 1.\]
Let $\alpha>0$ be as obtained in Section~\ref{sec:counter_eg_1}.
Consider any large $n\in \N$, and let $f:\set{0,1}^n\to [-1,1]$ be the function obtained in Section~\ref{sec:counter_eg_1}.

Let $g:\set{0,1}^n \to \set{-1,1}$ be a random function, defined as $g(x) = \begin{cases} 1, & w.p.\  \frac{1+f(x)}{2} \\ -1, & w.p.\  \frac{1-f(x)}{2} \end{cases}$, independently for each $x\in \set{0,1}^n$.
Observe that this satisfies $\E_{g}\sqbrac{g(x)} = f(x)$ for each $x\in \set{0,1}^n$.
We will show that the function $g$ satisfies the two desired properties with probability $1-o_n(1)$, and hence by the probabilistic method, this guarantees the existence of a non-random $g$ as desired.
This is done as follows:
\begin{enumerate}
	\item Let $F,G:\set{0,1}^{kn}\to [-1,1]$ be defined as $F(x) = \prod_{i\in [k]}f\brac{x_i}$ and $G(x) = \prod_{i\in [k]}g\brac{x_i}$.
		Let $X,Y \sim \nu^{\otimes n}$ be independent (of each other and of $g$) and let $E$ be the event that $X_1, \dots, X_k, Y_1, \dots, Y_k$ are all distinct.
		Then, by a union bound, we have that $\Pr\sqbrac{\bar E} \leq 2\cdot \binom{k}{2} \cdot \eta^n + k^2\cdot \brac{p^2+\brac{1-p}^2}^n = o_n(1)$, and hence
		\begin{align*}
			\abs{\ \E_{g} \E_{X\sim \nu^{\otimes n}}\sqbrac{G(X)}-\E_{X\sim \nu^{\otimes n}}\sqbrac{F(X)}\ } 
			&\leq \Pr\sqbrac{\bar E} + \abs{\  \E_{g}\E_{X,Y}\sqbrac{G(X)\cdot \ind_{E}}-\E_{X}\sqbrac{F(X)}\ }
			\\&\leq \Pr\sqbrac{\bar E} + \abs{\ \E_{X,Y}\sqbrac{F(X)\cdot \ind_{E}}-\E_{X}\sqbrac{F(X)}\ } 
			\\&\leq 2\Pr\sqbrac{\bar E} = o_n(1).
		\end{align*}
		Similarly, we have
		\begin{align*}
			\abs{\ \E_{g} \sqbrac{\E_{X}\sqbrac{G(X)}}^2-\sqbrac{\E_{X}\sqbrac{F(X)}}^2\ } &= \abs{\ \E_{g} \E_{X,Y}\sqbrac{G(X)\cdot G(Y)}-\E_{X,Y}\sqbrac{F(X)\cdot F(Y)}\ } \\&\leq 2\Pr\sqbrac{\bar E} = o_n(1).
		\end{align*}
		
		Letting $\beta = \abs{\E_{X}\sqbrac{F(X)}} \geq \alpha$, we get 
		$ \Var_g\sqbrac{\E_{X}\sqbrac{G(X)}} \leq \beta^2+o_n(1) - \brac{\beta-o_n(1)}^2 = o_n(1)$.
		Hence, by Chebyshev's inequality (Fact~\ref{fact:chebyshev}), we have $\abs{\E_{X}\sqbrac{G(X)}} \geq \frac{\alpha}{2}$ with probability $1-o_n(1)$.
		
	\item Fix $S \subseteq [n]$. Let $X\sim \mu_p^{\otimes n}$, and let $W = \E_{X}\sqbrac{\chi_S(X)\cdot g(X)} = \sum_{x\in \set{0,1}^n}\Pr\sqbrac{X=x} \cdot \chi_S(x)\cdot g(x)$. Observe that $W$ is a sum of $2^n$ independent and bounded random variables, and such that $\E_g[W] = \E_{X}\sqbrac{\chi_S(X)\cdot f(X)}$. For $q = \max\set{p,1-p}<1$, it holds that $\sum_x (2\Pr[X=x])^2 \leq 4q^n\cdot \sum_x\Pr[X=x] = 4q^n$, and by Hoeffding's inequality (Fact~\ref{fact:hoeffding}), we have for each $t>0$ that
	\[ \Pr\sqbrac{ \abs{W-\E[W]}\geq t} \leq 2\cdot \exp\brac{-\frac{2t^2}{4q^n}}. \]
	Let $t=q^{n/4}$.
	Then, with probability at least $1-o_n(2^{-n})$, it holds that $\abs{W} = \abs{\E_{X}\sqbrac{\chi_S(X)\cdot g(X)}} \leq \abs{\E_{X}\sqbrac{\chi_S(X)\cdot f(X)}}+q^{n/4} = o_n(1).$
	
	Now, a union bound over $S\subseteq[n]$ shows that with probability $1-o_n(1)$, the above holds for every $S\subseteq [n]$.
	\qed
\end{enumerate}

\section{Queries vs. Bias Tradeoff}\label{sec:query_bias}

In this section, we analyze the relation between $p$ (the bias) and $k$ (the number of queries) for the existence of a distribution $\nu \in \mc D(p,k)$ with some pairwise independent coordinate, and with full even-weight support (see Definition~\ref{defn:distr_class}).

\subsection{Query Lower Bound}

We prove a lower bound on $k$ in terms of the $p$, as follows:
\begin{proposition}\label{prop:query_bias_lb}
	Let $k\in \N,\ p\in (0,1)$, and let $\nu \in \mc D(p,k)$ be a distribution that has some pairwise independent coordinate.
	Then, it holds that $k\geq 3$ and  $\frac{1}{k-1} \leq p \leq 1-\frac{1}{k-1}$.
\end{proposition}
\begin{proof}
	Let $X\sim \nu$, and let $i\in [k]$ be a pairwise independent coordinate under $\nu$.
	
	For $Z = \sum_{j\not=i} X_j$, we have by linearity of expectation, that $\E\sqbrac{X_i\cdot Z} = (k-1)p^2$.
	On the other hand, observe that if $X_i = 1$, then $Z = 1 \modt$ and so $Z\geq 1$.
	Hence,
	\[p = \E\sqbrac{X_i\cdot 1}\leq  \E\sqbrac{X_i\cdot Z} = (k-1)p^2, \]
	and we have $(k-1)p \geq 1$; in particular, this shows $k\geq 3$.
	
	For the upper bound on $p$, we consider the following cases:
	\begin{itemize}
		\item $k$ is odd: In this case, if $X_i=1$, then $Z = 1 \modt$ and so $Z\leq k-2$. Hence,
		\[ (k-1)p^2 = \E\sqbrac{X_i\cdot Z} \leq  \E\sqbrac{X_i\cdot (k-2)} = p(k-2),\]
		and we have $(k-1)p \leq (k-2)$, as desired.
		\item $k$ is even: In this case, observe that the distribution of the random variable $(1-X_1, \dots, 1-X_k)$ also satisfies the hypothesis of the proposition, with $p$ replaced by $1-p$. Hence, the above proof gives us $(k-1)\cdot (1-p) \geq 1$, as desired. \qedhere
	\end{itemize}
\end{proof}
\begin{remark}\label{remark:corner_case_full_support}
The proof of Proposition~\ref{prop:query_bias_lb} also shows that for $k>3$ and $p\in \set{\frac{1}{k-1}, 1-\frac{1}{k-1}}$, any distribution satisfying the assumptions of Proposition~\ref{prop:query_bias_lb} cannot have full even-weight support.
This is because if $p\in \set{\frac{1}{k-1}, 1-\frac{1}{k-1}}$, in all cases in the above proof, the random variable $Z$ must be constant under some value of $X_i$ (either $X_i=0$ or $X_i=1$); this cannot be the case for a distribution with full even-weight support when $k>3$.
\end{remark}
%%%%%%%%%%%%%%%%%%%%%%%%%%%%%%%%%%%%%%%%%%%%%%%%%%%%%%%%%%%%%%%%%%%%%%%%%%%%%%%%

\subsection{Query Upper Bound}

In this subsection, we shall prove the following proposition.
\begin{proposition}\label{prop:query_bias_ub}
	Let $k\geq 3$ be a positive integer, and let $p \in \sqbrac{\frac{1}{k-1},1-\frac{1}{k-1}}$ (note that this interval is non-empty for $k\geq 3$).
	
	Then, there exists a permutation-invariant\footnote{we say that a distribution $\nu$ over $\set{0,1}^k$ is \emph{permutation-invariant}, if for $X = (X_1,\dots,X_k)\sim \nu$, and any permutation $\pi:[k]\to [k]$, the distribution of $\brac{X_{\pi(1)}, \dots, X_{\pi(k)}}$ is the same as $\nu$.} and pairwise independent distribution $\nu (k,p) \in \mc D(p,k)$ (see Definition~\ref{defn:distr_class}).
	Furthermore, if $k=3$ or if $p\not\in \set{\frac{1}{k-1}, 1-\frac{1}{k-1}}$, then there exists such a distribution with full even-weight support.
\end{proposition}

The proof involves various cases, considered below in Lemma~\ref{lemma:ub_case_analysis} and Lemma~\ref{lemma:ub_add_ind_copies}.

\begin{lemma}\label{lemma:ub_case_analysis}
	Let $k\geq 4$ be a positive integer, and let $p \in \left[\frac{1}{k-1}, \frac{2}{k-1}\right)\cup \left( 1-\frac{2}{k-1}, 1-\frac{1}{k-1}\right]$ (note that this interval is contained in $\sqbrac{\frac{1}{k-1},1-\frac{1}{k-1}}$ for $k\geq 4$). 
	Then, there exists a pairwise independent distribution $\nu (k,p)\in \mc D(p,k)$.
	
	Moreover, if $p\not\in \set{\frac{1}{k-1}, 1-\frac{1}{k-1}}$, then there exists such a distribution with full even-weight support.
\end{lemma}
\begin{proof}
	Let $k\geq 4$ be a positive integer, and let $p \in \left[\frac{1}{k-1}, \frac{2}{k-1}\right)\cup \left( 1-\frac{2}{k-1}, 1-\frac{1}{k-1}\right]$.
	Let $s = \lfloor\frac{k}{2}\rfloor$; we shall exhibit a vector  $q = (q_0,q_1,\dots,q_s) \in [0,1]^{s+1}$ satisfying:
	\[
		\sum_{i=0}^s \binom{k}{2i}\cdot  q_i = 1,\ \quad\sum_{i=1}^s \binom{k-1}{2i-1}\cdot q_i = p,\quad \sum_{i=1}^s \binom{k-2}{2i-2}\cdot q_i = p^2.
	\]
	The distribution $\nu(p,k)$ is then defined 	by assigning probability $\begin{cases} q_{\abs{x}/2},& \abs{x}=0\modt\\0,& \abs{x}=1\modt \end{cases}$ to the point $x\in \set{0,1}^k$, where $\abs{x} = \sum_{i=1}^k x_i$.
	Note that the above properties correspond to $\nu(k,p)$ being a valid probability distribution supported on even-hamming-weight vectors, having marginals $\mu_p$, and pairwise independent coordinates.
	Furthermore, the distribution $\nu(p,k)$ has full even-weight support if and only if each $q_i \in (0,1)$.
	
	The vector $q$ is defined as follows in different cases (for brevity, we omit the verification of the above properties):
	\begin{enumerate}
		\item $k\geq 5$ is odd,  $p \in \left[\frac{1}{k-1}, \frac{2}{k-1}\right)$: Let $q_0 = 1 + \frac{kp^2}{2} - \frac{k^2p}{2(k-1)}$, $q_1 = \frac{(k-2)p-(k-1)p^2}{(k-1)(k-3)}$, $q_{(k-1)/2} = \frac{(k-1)p^2-p}{(k-1)(k-3)}$, and zero otherwise.
		\item $k\geq 5$ is odd,  $1-p\in \left[\frac{1}{k-1}, \frac{2}{k-1}\right)$: Let $q_0 = 1 + \frac{kp^2}{k-3} - \frac{k(2k-5)p}{(k-1)(k-3)}$, $q_{(k-3)/2} = \frac{3(k-2)p-3(k-1)p^2}{(k-1)(k-2)(k-3)}$, $q_{(k-1)/2} = \frac{(k-1)p^2-(k-4)p}{2(k-1)}$, and zero otherwise.
		\item $k\geq 4$ is even,  $p\in \left[\frac{1}{k-1}, \frac{2}{k-1}\right)$: Let $q_0 = \frac{(k-1)p^2-(k+1)p+2}{2}$,  $q_1 = \frac{p-p^2}{k-2}$, $q_{k/2} = \frac{(k-1)p^2-p}{k-2}$, and zero otherwise.
		\item $k\geq 4$ is even, $1-p\in \left[\frac{1}{k-1}, \frac{2}{k-1}\right)$: In this case, we define $\nu(k,p)$ to be the distribution obtained by flipping each coordinate of $\nu(k,1-p)$.
	\end{enumerate}
	
	Next, we show that if $p\not\in \set{\frac{1}{k-1}, 1-\frac{1}{k-1}}$, then such a distribution $\nu(p,k)$ with full even-weight support exists.
	We only need to do this for the first three cases, as the procedure described in the fourth case preserves the property of full even-weight support.
	 
	The same argument applies in all cases, and we present it for the first case: that is when $k\geq 5$ is odd, and $p \in \brac{\frac{1}{k-1}, \frac{2}{k-1}}$.
	We observe if $p \not= \frac{1}{k-1}$, each of the probabilities $q_0,q_1,q_{(k-1)/2}$ above lie in the interval $(0,1)$.
	Now, consider the equations
	\[
		\sum_{i=0}^s \binom{k}{2i}\cdot  \tilde{q}_i = 0,\ \quad\sum_{i=1}^s \binom{k-1}{2i-1}\cdot \tilde{q}_i = 0,\quad \sum_{i=1}^s \binom{k-2}{2i-2}\cdot \tilde{q}_i = 0.
	\]
	In these equations, the variables $\tilde{q}_0, \tilde{q}_1, \tilde{q}_{(k-1)/2}$ are linearly independent, and hence, there exists a vector $\tilde{q} \in \R^{s+1}$ satisfying these equations, which has all coordinates equal to 1, other than possibly $\tilde{q}_0, \tilde{q}_1, \tilde{q}_{(k-1)/2}$.
	Then, for some small $\delta > 0$, the vector $q+\delta\cdot \tilde{q}$ has all coordinates in $(0,1)$, and satisfies the required properties.	
\end{proof}

\begin{lemma}\label{lemma:ub_add_ind_copies}
	Let $k\geq 6$ be a positive integer, and let $p \in \sqbrac{\frac{2}{k-1},1-\frac{2}{k-1}} \setminus \set{\frac{1}{2}}$ (note that this interval is non-empty for $k\geq 6$).
	There, there exists a pairwise independent distribution $\nu (k,p)\in \mc D(p,k)$ with full even-weight support.
\end{lemma}
\begin{proof}
	Let $k\geq 6$ be a positive integer, and let $p \in \sqbrac{\frac{2}{k-1},1-\frac{2}{k-1}},\ p\not=\frac{1}{2}$.
	That is, for $q = \min\set{p,1-p} < \frac{1}{2}$, we have $k \geq 1 + \frac{2}{q}$.
	Let $\ell$ be the smallest odd integer satisfying $\ell > 1+\frac{1}{q} > 3$.
	Note that this satisfies $4\leq \ell \leq 3+\frac{1}{q} < 1 + \frac{2}{q} \leq k$, and we have $q \in \brac{\frac{1}{\ell-1}, \frac{2}{\ell-1}}$.	
	
	By Lemma~\ref{lemma:ub_case_analysis}, there exist pairwise independent distributions $\nu(\ell, p)$ and $\nu(\ell, 1-p)$, with full even-weight support.
	Let $\tilde{\nu}_0 = \nu(\ell, p)$, and let $\tilde{\nu}_1$ be the distribution obtained by flipping each coordinate of $\nu(\ell, 1-p)$.
	Since $\ell$ is odd, for each $b\in \set{0,1}$, it holds that $\tilde{\nu}_b$ has pairwise independent coordinates, each with marginal $\mu_p$, and such that $\supp(\tilde{\nu}_b) = \set{x\in \set{0,1}^k: \sum_{i=1}^k x_i = b \modt}$.
	Finally, we define $X \sim \nu(k,p)$ via the following random process: Let $(X_{\ell+1},\dots,X_k)\sim \mu_p^{\otimes\brac{k-\ell}}$, and with $Z = \sum_{i=\ell+1}^k X_i \modt$, we let $(X_1,\dots,X_\ell) \sim \tilde{\nu}_{Z}$.
	It is an easy check that this distribution satisfies the required properties.
\end{proof}

Finally, we prove Proposition~\ref{prop:query_bias_ub}.

\begin{proof}[Proof of Proposition~\ref{prop:query_bias_ub}]
	Note that it suffices to find such a distribution that is not necessarily permutation invariant, since averaging the distribution over all permutations preserves pairwise independence and full even-weight support.
	
	If $p=1/2$, for any $k\geq 3$, we let $\nu(k,p)$ be the uniform distribution on the set $\set{ x\in\set{0,1}^k : \sum_{i=1}^k x_i = 0 \modt}$.
	
	Now, for $k=3$, it must hold that $p=1/2$, in which case $\nu(k,p)$ is as above.
	For $k=4$ or $k=5$, and $p\not=1/2$, it must hold that $p \in \left[\frac{1}{k-1}, \frac{2}{k-1}\right)\cup \left( 1-\frac{2}{k-1}, 1-\frac{1}{k-1}\right]$, and the result follows from Lemma~\ref{lemma:ub_case_analysis}.
	For $k\geq 6$ and $p\not=\frac{1}{2}$, the result follows from Lemma~\ref{lemma:ub_case_analysis} and Lemma~\ref{lemma:ub_add_ind_copies}.
\end{proof}

%%%%%%%%%%%%%%%%%%%%%%%%%%%%%%%%%%%%%%%%%%%%%%%%%%%%%%%%%%%%%%%%%%%%%%%%%%%%%%%%

\section{Putting Everything Together}\label{sec:putting_together}

We are now ready to prove our main result.

\begin{proof}[Proof of Theorem~\ref{thm:intro_querybias_main_thm}]
Let $p\in (0,1)$.

\begin{enumerate}
	\item Consider any positive integer $k > 1 + \frac{1}{\min\set{p,1-p}} \geq 3$ (or $k=3$ with $p=\frac{1}{2}$).
		By Proposition~\ref{prop:query_bias_ub}, there exists a pairwise independent distribution $\nu\in \mc D(p,k)$ with full even-weight support.
		The result now follows by Theorem~\ref{thm:bkm23_in_section}.
	
	\item Suppose that $k\geq 3$ with $p=\frac{1}{k-1}$, or $k\geq 4$ is even with $\ p = 1-\frac{1}{k-1}$.
	In these cases, we observe that the distribution $\nu\in \mc D(p,k)$ constructed in Lemma~\ref{lemma:ub_case_analysis} is pairwise independent, and contains BLR (see Definition~\ref{defn:cont_BLR}):
	\begin{enumerate}
		\item If $k\geq 3, p=\frac{1}{k-1}$, the distribution $\nu$ contains all vectors in $\set{0,1}^k$ of hamming-weights 0 and 2 in its support. In this case, Definition~\ref{defn:cont_BLR} is satisfied with $\tilde{b}=0$ and $\tilde{z}$ as the all-zeros vector.
		\item If $k\geq 4$ is even, and $p=1-\frac{1}{k-1}$, the distribution $\nu$ contains all vectors in $\set{0,1}^k$ of hamming-weights $k-2$ and $k$ in its support. In this case, Definition~\ref{defn:cont_BLR} is satisfied with $\tilde{b}=1$ and $\tilde{z}$ as the all-ones vector.
	\end{enumerate}
	The result now follows by Theorem~\ref{thm:bkm23_in_section}.
	
	\item Suppose that $k < 1 + \frac{1}{\min\set{p,1-p}}$ is a positive integer, and let $\nu \in \mc D(p,k)$.
		We perform the following operation on the distribution $\nu$:
		if $i,j\in [k],\ i\not=j$ are such that $\Pr_{X\sim \nu}[X_i=X_j]=1$, we remove coordinates $i,j$ from $\nu$, and repeat until no such pairs remain.
		
		Finally, we are left with a distribution $\tilde{\nu}$ on $\tilde{k}\leq k$ coordinates.		
		We consider the following two cases:
		\begin{enumerate}
			\item Suppose that $\tilde{k}=0$. In this case, for every $n\in \N,$ and every $f:\set{0,1}^n\to \set{-1,1}$, it holds that $\E_{X\sim \nu^{\otimes n}} \sqbrac{\prod_{i=1}^kf(X_i)} = 1$, since the $k$ terms in the product cancel out in pairs.
			Hence, it suffices to show the existence of a function $f:\set{0,1}^n\to \set{-1,1}$ satisfying $\abs{\E_{X\sim \mu_p^{\otimes n}}\sqbrac{f(X)\cdot \chi_S(X)}} \leq o_n(1)$ for every $S\subseteq [n]$.
			Note that a (uniformly) random function $f:\set{0,1}^n\to \set{-1,1}$ satisfies this with high probability, by an argument similar to the one at the end of Section~\ref{sec:counter_eg_2} (a random function can be thought of as rounding the constant zero function as in Section~\ref{sec:counter_eg_2}).
			
			\item Now, suppose that $\tilde{k}\not=0$. 
			Then, it holds that $\tilde{\nu} \in \mc D(p,\tilde{k})$, and by Proposition~\ref{prop:query_bias_lb}, we have that $\tilde{\nu}$ has no pairwise independent coordinate.
			Now, by Theorem~\ref{thm:counter_eg} there exists a constant $\alpha>0$, such that for every large $n\in \N$, there exists a function $f:\set{0,1}^n\to \set{-1,1}$ such that \[ \abs{\E_{(X_1,\dots,X_k)\sim \nu^{\otimes n}} \sqbrac{\prod_{i\in [k]} f(X_i)} } = \abs{\E_{(X_1,\dots,X_{\tilde{k}})\sim \tilde{\nu}^{\otimes n}} \sqbrac{\prod_{i\in [\tilde{k}]} f(X_i)} }\geq \alpha,\]
			and such that $ \abs {\E_{X\sim \mu_p^{\otimes n}}\sqbrac{f(X)\cdot \chi_S(X)}} \leq o_n(1)$ for every $S\subseteq [n]$. \qedhere
			\end{enumerate}
\end{enumerate}
\end{proof}

%%%%%%%%%%%%%%%%%%%%%%%%%%%%%%%%%%%%%%%%%%%%%%%%%%%%%%%%%
\subsection{A Corner Case}\label{sec:corner_case}

In the above proof, we leave the case of odd $k\geq 5$ and $p=1-\frac{1}{k-1}$.
This turns out to be very interesting, and we discuss it next.
For the remainder of this section, we fix such a $k$ and $p$.

In this case, the pairwise independent distribution $\nu\in \mc D(p,k)$ constructed in Lemma~\ref{lemma:ub_case_analysis}, is supported on vectors of hamming weights 0 and $k-1$ (and does not contain BLR as in Definition~\ref{defn:cont_BLR}).
In particular, for every $x\in \supp(\nu)$, it holds that $\sum_{i=1}^k x_i = 0 \modk$.
For this reason, as we show next, the best we can expect from the test $\Lin(\nu)$, is to guarantee correlation with a character over $\Z/(k-1)\Z$, and this is indeed true.

\begin{definition}\label{defn:char}(Characters over $\Z/(k-1)\Z$)
Let $\omega$ be a primitive $(k-1)$\textsuperscript{th} root of unity.
For every $0\leq r\leq k-2$, we define the function $\phi_r: \set{0,1} \to \C$ as $\phi_r(x)=\omega^{rx}$.

For every $n\in \N$, and every integers $0\leq r\up{1},\dots,r\up{n}\leq k-2$, we define the product character $\phi_{r\up{1},\dots,r\up{n}}:\set{0,1}^n\to \C$ by $\phi_{r\up{1},\dots,r\up{n}}(x) = \prod_{j=1}^n\phi_{r\up{j}}(x\up{j}) = \omega^{\sum_{j=1}^n r\up{j}x\up{j}}$.
\end{definition} 

Now, consider the test $\Lin(\nu)$.
Observe that any character $f = \phi_{r\up{1},\dots,r\up{n}}$ passes this test with probability 1:
\[ \E_{X\sim\nu^{\otimes n}}\sqbrac{\prod_{i\in [k]}f(X_i)} = \prod_{j=1}^n \E_{Y\sim\nu}\sqbrac{\prod_{i\in [k]}\phi_{r\up{j}}(Y_i)} = \prod_{j=1}^n \E_{Y\sim\nu}\sqbrac{\omega^{r\up{j}\cdot \brac{\sum_{i\in [k]}{Y_i}}}} = 1.\]
Next, we claim that characters explain the success of $\Lin(\nu)$ for any function $f$:

\begin{theorem}
	For every constant $\epsilon > 0$, there exists a constant $\delta >0$ such that for every large enough $n\in \N$, the following is true:
	
	Let $f:\set{0,1}^n\to[-1,1]$ be a function such that $\abs{\E_{X\sim \nu^{\otimes n}} \sqbrac{\prod_{i=1}^k f(X_i)} }\geq \eps.$
	Then, there exist integers $0\leq r\up{1},\dots,r\up{n}\leq k-2$, such that
	\[ \abs {\E_{X\sim \mu_p^{\otimes n}}\sqbrac{f(X)\cdot \phi_{r\up{1},\dots,r\up{n}}(X)}} \geq \delta .\]
\end{theorem}
\begin{proof}
	The result follows from the work of Bhangale, Khot, Liu and Minzer~\cite{BKLM24a, BKLM24b}, and we omit the details.
	Very roughly speaking, the proof follows a similar strategy as in Section~\ref{sec:bkm_sketch}: first show that $f$ has good correlation with a character under random restrictions; then, use this to show that $f$ has good correlation with character times a low-degree function; finally, use that $\nu$ is pairwise independent to get rid of the low-degree function.
\end{proof}

Finally, we present an alternative solution to deal with this corner case of odd $k\geq 5$ and $p=1-\frac{1}{k-1}$.
Instead of the test $\Lin(\nu)$, we can perform the following test:

Let $f:\set{0,1}^n\to [-1,1]$, and let $\nu' \in \mc D(1-p,k)=\mc D(\frac{1}{k-1},k)$ be the pairwise independent distribution from Lemma~\ref{lemma:ub_case_analysis}.
\begin{enumerate}
	\item Sample $X=(X_1,\dots,X_k)\sim \nu'^{\otimes n}$.
	\item Let $X'$ be the vector obtained by negating each of the $kn$ coordinates of $X$.
	\item Query $f$ on $X'_1,\dots,X'_k$ and accept if and only if $\prod_{i\in [k]} f(X'_i) = 1$.  
\end{enumerate}
Each query $X_i'$ of the above test is distributed according to $\mu_p^{\otimes n}$, and the analysis of the test simply follows from the analysis for $\Lin(\nu')$ in Theorem~\ref{thm:intro_querybias_main_thm}.
The drawback here, though, is that the test does not accept all linear functions with probability 1, but only functions of the form $(-1)^{\abs{S}}\cdot \chi_S,$ for $S\subseteq[n]$.

\section{Analysis of the Linearity Test}\label{sec:bkm_sketch}

In this section, we shall state and prove a generalized version of Theorem~\ref{thm:bkm23}.
The proof follows the work of Bhangale, Khot and Minzer~\cite{BKM23b}, and hence we only give a rough outline (skipping many of the technical points), pointing out the places where the proof differs from the above work.
We start with the following definition:

\begin{definition}\label{defn:cont_BLR}
	Let $k\geq 3,p\in (0,1)$, and let $\nu\in \mc D(p,k)$ be a distribution.
	We say that $\nu$ \emph{contains BLR}, if there exists some $\tilde{b}\in \set{0,1},\ \tilde{z}\in \set{0,1}^{k-3}$, such that 
	\[ \set{(x_1,\ x_2,\ x_1\oplus x_2\oplus \tilde{b},\ \tilde{z}) : x_1,x_2\in \set{0,1}}\subseteq \supp(\nu)\subseteq \set{0,1}^k. \]	
	Furthermore, for technical reasons, we shall also require that \[\textnormal{span}_{\mathbb{F}_2}(\supp(\nu)) = \set{x\in \set{0,1}^k : \sum_{i=1}^k x_i = 0\modt} .\]
\end{definition}

Observe that any $\nu$ with full even-weight support contains BLR (with $\tilde{b}=0$, and $\tilde{z}$ the all-zeros vector).
With this, we state the following generalization of Theorem~\ref{thm:bkm23}:

\begin{theorem}\label{thm:bkm23_in_section}
	Let $k\geq 3$ be a positive integer, and let $p\in (0,1),\ \epsilon \in (0,1]$ be constants, and let $\nu \in \mc D(p,k)$ be a distribution containing BLR (see Definition~\ref{defn:cont_BLR}).
	Then, there exists constants $\delta>0,\ d\in \N$ (possibly depending on $k, p, \epsilon, \nu$), such that for every large enough $n\in \N$, the following is true:
	
	Let $f:\set{0,1}^n\to[-1,1]$ be a function such that \[ \abs{\E_{(X_1,\dots,X_k)\sim \nu^{\otimes n}} \sqbrac{\prod_{i=1}^k f(X_i)} }\geq \eps.\]
	Then, there exists a set $S\subseteq [n]$, and a polynomial $g:\set{0,1}^n\to \R$ of degree at most $d$ and with 2-norm $\E_{X\sim \mu_p^{\otimes n} }\sqbrac{g(X)^2}\leq 1$, such that
	\[ \abs {\E_{X\sim \mu_p^{\otimes n}}\sqbrac{f(X)\cdot \chi_S(X)\cdot g(X)}} \geq \delta .\]
	
	Moreover, if the distribution $\nu$ has some pairwise independent coordinate, then we may assume $g\equiv 1$; that is, $f$ correlates with a linear function $\chi_S$.
\end{theorem}

The remainder of this section is devoted to the proof of the above theorem.
Let $k\geq 3$ be an integer, and let $\ p\in (0,1),\ \epsilon \in (0,1]$ be constants, and let $\nu \in \mc D(p,k)$ be a distribution containing BLR (see Definition~\ref{defn:cont_BLR}).
Also, let $f:\set{0,1}^n\to[-1,1]$ be a function such that

\begin{equation}\label{eqn:test_pass}
	\abs{\E_{X=(X_1,\dots,X_k)\sim \nu^{\otimes n}} \sqbrac{\prod_{i=1}^k f(X_i)} }\geq \eps.
\end{equation}

\subsection*{Step 1: Large Fourier Coefficient under Random Restriction.}\label{sec:large_fcurr}
We note that the proof of this step is where we differ from~\cite{BKM23b}.

Since the distribution $\nu\in \mc D(p,k)$ contains BLR, we can write $\nu = (1-\beta)\cdot \nu' + \beta\cdot \mu$, for some small constant $0<\beta<\frac{1}{2}\min\set{p,1-p}$, some distribution $\nu'$ over $\set{0,1}^k$, and with $\mu$ the uniform distribution over $\set{(x_1,x_2,x_1\oplus x_2\oplus \tilde{b},\tilde{z}) : x_1,x_2\in \set{0,1}}$, where $\tilde{b},\tilde{z}$ are as in Definition~\ref{defn:cont_BLR}.
Using this, we can describe choosing $X \sim \nu^{\otimes n}$ as the following two step process. First choose a set $I\subseteq [n]$, denoted $I\sim_{1-\beta} [n]$, by choosing $i\in I$ with probability $1-\beta$, independently for each $i\in [n]$.
Then, choose $Z\sim \nu'^{\otimes I}$ and $Y\sim \mu^{\bar{I}}$, and set $X = (Y,Z)$.

With the above, we can prove that the function $f$ satisfies the property of having a large fourier coefficient under random restrictions; the reader is referred to~\cite{Don14} for an introduction to Fourier analysis over the hypercube.

\begin{lemma}\label{lemma:lfcurr}
	With $\delta = \epsilon/2$, it holds that
	\[ \Pr_{I\sim_{1-\beta}[n],\ Z\sim \nu'^{\otimes I}}\sqbrac{\exists S\subseteq [n]\setminus I:\ \abs{\widehat{f_{I\to Z_1}}(S)}\geq \delta\ } \geq \delta.\]
	Here, $f_{I\to Z_1}$ refers to the restriction of the function $f$, with the variables in $I$ \emph{set to} $Z_1$.
\end{lemma}	
\begin{proof}
	By Equation~\ref{eqn:test_pass}, we have
	\begin{align*}
		\epsilon &\leq \abs{\ \E_{X=(X_1,\dots,X_k)\sim \nu^{\otimes n}} \sqbrac{\prod_{i=1}^k f(X_i)}\ }
		\\&= \abs{\ \E_{I\sim_{1-\beta}[n],\ Z\sim \nu'^{\otimes I}}\E_{Y\sim \mu^{\otimes \bar{I}}} \sqbrac{\prod_{i=1}^k f_{I\to Z_i}(Y_i)}\ }
		\\&\leq \E_{I\sim_{1-\beta}[n],\ Z\sim \nu'^{\otimes I}}\abs{\ \E_{Y\sim \mu^{\otimes \bar{I}}} \sqbrac{\prod_{i=1}^k f_{I\to Z_i}(Y_i)}\ }		
	\end{align*}
	Observe that in the above expression, the random variables $Y_4,\dots,Y_k$ are constants (determined by $\tilde{z}$).
	Now, using a (classical) Fourier analytic argument to analyze the BLR linearity test over the uniform distribution (see Chapter 1 of~\cite{Don14}), we get
	\begin{align*}
		\epsilon &\leq \E_{I\sim_{1-\beta}[n],\ Z\sim \nu'^{\otimes I}}\abs{\ \E_{Y\sim \mu^{\otimes \bar{I}}} \sqbrac{\prod_{i=1}^3 f_{I\to Z_i}(Y_i)}\ }		
%		\\&= \E_{I\sim_{1-\beta}[n],\ Z\sim \nu'^{\otimes I}}\abs{\ \E_{Y_1,Y_2\sim \set{0,1}^{\otimes \bar{I}}} \sqbrac{ f_{I\to Z_1}(Y_1)\cdot f_{I\to Z_2}(Y_2)\cdot f_{I\to Z_3}(Y_1\oplus Y_2\oplus \tilde{b}^{\bar{I}}) }\ }
		\\&= \E_{I\sim_{1-\beta}[n],\ Z\sim \nu'^{\otimes I}}\abs{\ \sum_{S\subseteq \bar{I}} \widehat{f_{I\to Z_1}}(S)\cdot \widehat{f_{I\to Z_2}}(S)\cdot \widehat{f_{I\to Z_3}}(S)\cdot (-1)^{\tilde{b}\cdot \abs{S}} \ }
		\\&\leq \E_{I\sim_{1-\beta}[n],\ Z\sim \nu'^{\otimes I}} \sqbrac{\max_{S\subseteq \bar{I}}\abs{\widehat{f_{I\to Z_1}}(S)}}
		\\&\leq \Pr_{I\sim_{1-\beta}[n],\ Z\sim \nu'^{\otimes I}}\sqbrac{\exists S\subseteq \bar{I}:\ \abs{\widehat{f_{I\to Z_1}}(S)}\geq \epsilon/2}  + \epsilon/2.
		\qedhere
	\end{align*}
\end{proof}	

\subsection*{Step 2: Direct Product Test}
Using Theorem 1.1 in~\cite{BKM23b}, by Lemma~\ref{lemma:lfcurr} we get the existence of constants $d\in \N, \delta'>0$, a set $S\subseteq [n]$, and a polynomial $g:\set{0,1}^n\to \R$ of degree at most $d$, and with 2-norm $\E_{X\sim \mu_p^{\otimes n} }\sqbrac{g(X)^2}\leq 1$, such that \[ \abs {\E_{X\sim \mu_p^{\otimes n}}\sqbrac{f(X)\cdot \chi_S(X)\cdot g(X)}} \geq \delta' .\]
	
This proves the first part of Theorem~\ref{thm:bkm23}.
It remains to show that if $\nu$ has some pairwise independent coordinate, it is possible to remove the function $g$ in the above expression.
	
\subsection*{Step 3: List Decoding.} This step follows Section 4.2 and Section 4.3 in~\cite{BKM23b}.

Using an iterative list-decoding process, we can find a constant $r\in \N$, and functions $\chi_{S_1}, \dots, \chi_{S_r}$, and constant degree polynomials $g_1,\dots, g_r,$ such that it is possible to ``replace" $f$ by $\sum_{i\in [r]}\chi_{S_i}\cdot g_i$ in Equation~\ref{eqn:test_pass} (and lose at most some constant factor in $\epsilon$).
Now, this implies that for some constant $\epsilon'>0$, and some indices $j_1,\dots,j_k\in [r]$, we have  
\begin{equation}\label{eqn:after_list}
	\abs{\E_{(X_1,\dots,X_k)\sim \nu^{\otimes n}} \sqbrac{\prod_{i=1}^k \chi_{S_{j_i}}(X_i) g_{j_i}(X_i) } }\geq \epsilon'.
\end{equation}
We remark that for the next step, some extra structure on $S_{j_i}$'s is needed, and ensuring that it holds requires the condition on $\textnormal{span}_{\F_2}(\supp(\nu))$ in Definition~\ref{defn:cont_BLR}.
	
\subsection*{Step 4: Invariance Principle Argument.}
This step follows Section 4.4, Section 4.5, and Section 4.6 in~\cite{BKM23b}. 

Assume, for the sake of contradiction, that $f$ is not correlated well with any $\chi_S$; that is, $\E_{X\sim \mu_p^{\otimes n}} \sqbrac{f(X)\cdot \chi_S(X)} \leq o_n(1)$ for each $S\subseteq [n]$.
Using this, it can be shown, roughly, that for each $i\in [k]$, the expectation $\E_{X\sim \mu_p^{\otimes n}}\sqbrac{\chi_{S_{j_i}}(X) g_{j_i}(X)} \leq o_n(1)$; note that for this conclusion to hold, we might have to modify $S_{j_i}$'s and $g_{j_i}$'s, however it is possible to do so while maintaining Equation~\ref{eqn:after_list}.

Now, by an invariance principle argument~\cite{MOO10, Mos10, Mos20}, very roughly, it is possible to replace the expectation in Equation~\ref{eqn:after_list} over $(X_1,\dots,X_k)\sim \nu^{\otimes n}$, by an expectation over $(Z_1,\dots,Z_k) \sim \mc N(0,\Sigma)^{\otimes n}$, where $\Sigma\in \R^{k\times k}$ is the (normalized) covariance matrix of $\nu$.
Finally, we use that some coordinate $X_{i^*}$ is pairwise independent of each $X_i$, for $i\not=i^*$.
Since the Gaussian distribution is determined by its covariance matrix, this implies that $Z_{i^*}$ is mutually independent of $(Z_i)_{i\not=i^*}$.
We have
\begin{align*}
	\epsilon' &\leq \abs{\E_{X=(X_1,\dots,X_k)\sim \nu^{\otimes n}} \sqbrac{\prod_{i=1}^k \chi_{S_{j_i}}(X_i) g_{j_i}(X_i) } }
	\\ &\approx  \abs{\E_{Z=(Z_1,\dots,Z_k)\sim \mc N(0,\Sigma)^{\otimes n}} \sqbrac{\prod_{i=1}^k \chi_{S_{j_i}}(Z_i) g_{j_i}(Z_i) } }
	\\&\approx \abs{\ \E_{Z_{i^*}\sim \mc N(0,1)^{\otimes n}} \sqbrac{ \chi_{S_{j_{i^*}}}(Z_{i^*}) g_{j_{i^*}}(Z_{i^*})} }\cdot  \abs{\ \E_{Z} \sqbrac{\prod_{i\in [k], i\not=i^*} \chi_{S_{j_i}}(Z_i) g_{j_i}(Z_i) } }
	\\&\approx \abs{\ \E_{X_{i^*}\sim \mu_p^{\otimes n}} \sqbrac{ \chi_{S_{j_{i^*}}}(X_{i^*}) g_{j_{i^*}}(X_{i^*})} }\cdot  \abs{\ \E_{Z} \sqbrac{\prod_{i\in [k], i\not=i^*} \chi_{S_{j_i}}(Z_i) g_{j_i}(Z_i) } }
	\\ &\leq o_n(1),
\end{align*}
which is a contradiction.
\qed

\section*{Acknowledgements}
We thank Amey Bhangale, Yang P. Liu, and Dor Minzer for discussions that helped this project.
Amey and Dor politely declined to be co-authors.

\bibliographystyle{alpha}
\bibliography{main.bib}

\end{document}